\documentclass[runningheads]{llncs}
\usepackage[T1]{fontenc}
\usepackage{amsmath, amssymb}
\usepackage{mathtools}
\usepackage{comment}

\usepackage{tikz}
\usetikzlibrary{fit,backgrounds, arrows.meta, calc, positioning, matrix}

\usepackage{algorithm}
\usepackage{algpseudocode}
\usepackage{placeins} 
\usepackage{url}
\newcommand{\doi}[1]{\url{https://doi.org/#1}}

\providecommand{\qedsymbol}{\ensuremath{\square}}

\newcommand{\qedbox}{\hfill\qedsymbol}

\title{Strong Conflict-Free Vertex-Connection via Twin Cover: Kernelization and Chromatic Bounds}

\titlerunning{Strong Conflict-Free Vertex-Connection via Twin Cover}

\usepackage{bbding}

\author{Samuel German \Envelope}
\authorrunning{Samuel German}
\institute{University of California, San Diego, USA \\
\email{sgerman@ucsd.edu}}

\begin{document}
\maketitle

\begin{abstract} A strong conflict-free vertex connection coloring of a graph is a graph coloring for which, for every two distinct vertices of the graph, there exists a shortest path between them in which some color in said coloring appears only once. The strong conflict-free vertex connection number, denoted by $\operatorname{svcfc}(G)$, is the number of colors in the strong conflict-free vertex connection coloring with the smallest amount of colors out of all such colorings of $G$. We study the problem of deciding whether this number is at most some threshold together with the parameter twin cover.  Our first result is that when we have  given $(G,k)$ in which $k$ is some target number of colors and given some twin cover $X$ of $G$ which is of size $t$, there is a polynomial time reduction to an equivalent instance which has at most $\max\{2,t+k2^t\}$ vertices. This implies that the annotated version of  the Strong CFVC Number problem, which is the same problem as  Strong CFVC Number problem except that a twin cover is supplied as part of the input, is FPT with respect to $t + k$.

This result allows us to give a kernel for Strong CFVC Number which is parameterized by $\operatorname{tc}(G)+k$, which implies that Strong CFVC Number is FPT w.r.t. the twin-cover number of $G$ so long as $k$ is fixed.

Our second result is that if $G$ is a connected graph for which there exists a twin cover of size $t$, then $\chi(G)\le \operatorname{svcfc}(G)\le \chi(G)+t$. A broader statement for this is that if $X$ is the name of the size-$t$ twin cover of $G$ whose existence we just mentioned, and if also $Y\subseteq X$ contains a vertex for every shortest path in $G$ whose length is at least $3$, then we have  $\operatorname{svcfc}(G)\le \chi(G)+|Y|$. 

Our last result is an an exact expression for the chromatic number of graphs whose twin cover number is bounded. Namely, we show that for every proper coloring $\varphi$ of $G[X]$, the minimum number of colors needed to extend $\varphi$ to all of $G$ (in that we would get a new coloring on domain $V(G)$ which agrees with $\varphi$ on the latter's domain) is $K_\varphi=\max_{S\subseteq X}(|\varphi(S)|+m(S))$, and hence $\chi(G)=\min_{\varphi\text{ proper on }G[X]} K_\varphi$. Our results indicate that twin cover is a meaningful parameter in the context of strong conflict-free vertex-connection, and show that once we have a fixed twin cover, the only difficulty left is in some additive threshold above the chromatic number. 

\end{abstract}

\section{Introduction}

Conflict-free coloring is a type of coloring in which each of a specified set of elements contains a uniquely colored element (``conflict" refers to another element having the same color as it, and thus being free of such conflicts denotes that there isn't an object with the same color as it). This concept originated in frequency assignments in cellular networks, and now it is a substantial area in the areas of combinatorics and algorithms~\cite{chang2024survey,even2003conflictfree,pach2009conflictfree,smorodinsky2013conflictfree}. A lot of areas involving graphs have now been enriched by having conflict-free notions introduced in them, such as  neighborhoods, paths, and matters of connectivity such as in conflict-free connection and strong conflict-free connectivity (see \cite{chang2024survey,czap2018conflictfree,ji2020strong}). One such natural connection involves the combination of the concept of the proper vertex coloring with constraints for shortest-path connectivity. Just that has been done in
 Hsieh et al. introducing the \emph{strong conflict-free vertex-connection} problem \cite{hsieh2025complexity}. First of all, some definitions in order to talk about the problem. A vertex-coloring of a connected graph $G$ is a strong conflict-free vertex-connection coloring if and only if for any two distinct vertices in $G$, there exists some shortest path (some path realizing the shortest distance between any two vertices in $G$) between the vertices for which there is some color such that only a single vertex on said shortest path is colored that color. The amount of colors in the coloring with this property that has the smallest amount of colors out of all such colorings is called $\operatorname{svcfc}(G)$.  And the problem is to determine whether this number is at most some threshold.

What we already know is that the strong conflict-free vertex-connection problem exhibits behavior of quite a different nature from that of classical graph coloring: among other results, it has been shown that  \textsc{Strong CFVC Number} is NP-hard already for the $3$-color case on some highly restricted graphs for which classical $3$-coloring can be decided in polynomial time\cite{hsieh2025complexity}.

As far as parameterized complexity, we have only begun to understand the strong conflict-free vertex-connection problem under this lens. It was recently shown by Feghali, H.O. Le, and V.B. Le that \textsc{Strong CFVC Number} is fixed-parameter tractable when parameterized by the vertex-cover number, as well as that the $3$-color case admits no polynomial kernel under this parameter unless $\mathsf{NP}\subseteq\mathsf{coNP/poly}$ \cite{feghali2026parameterized}. This result suggests that we need to consider parameters beyond vertex cover. One such canonical candidate is twin cover, which was introduced by Ganian \cite{ganian2015improving}. Twin cover strictly generalizes vertex cover. If $X$ is a twin cover of
$G$, then every component of $G-X$ is a clique of pairwise true twins
in $G$. The set $X$ can be small even when $G$ is dense. This decomposition
is useful for strong CFVC because proper colorings assign distinct colors
within each twin-clique, and every shortest path of length at least two
uses at most one vertex from each twin-clique. From the perspective of coloring itself, this choice is also natural given that twin cover has already proved useful for kernelization in fixed-target graph coloring problems \cite{jansen2019optimal}.

Our paper takes the first step in the study of twin cover for strong CFVC. We first of all shown an explicit kernel bound in the annotated setting in which a twin cover is supplied. More precisely, given a valid annotated instance $(G,k,X)$ with $X\subseteq V(G)$ a twin cover of size $t$, we compute in polynomial time an equivalent valid annotated instance $(G',k',X')$ such that
$|V(G')|\le \max\!\left\{2,\;t+k\,2^t\right\}$.
Thus the annotated version of \textsc{Strong CFVC Number} is fixed-parameter tractable parameterized by $t+k$. We then convert this into an unannotated result parameterized by $\operatorname{tc}(G)+k$: using a twin cover of size at most $2\,\operatorname{tc}(G)$, we obtain a kernel on at most
\[
\max\!\left\{
2,\;
2\,\operatorname{tc}(G)
+
k\,2^{2\,\operatorname{tc}(G)}
\right\}
\]
vertices.  In particular, for every fixed constant $k$, \textsc{Strong CFVC Number} is fixed-parameter tractable parameterized by just the twin-cover number.

Our next contribution is of a structural nature: we show that for every connected graph $G$ with twin cover $X$ of size $t$, $\chi(G)\leq \operatorname{svcfc}(G)\leq \chi(G)+t$.
 A broader statement for this is that if $X$ is the name of the size-$t$ twin cover of $G$ whose existence we just mentioned, and if also $Y\subseteq X$ contains a vertex for every shortest path in $G$ whose length is at least $3$, then we have  $\operatorname{svcfc}(G)\le \chi(G)+|Y|$. 

To complement this, we give an exact expression for the chromatic number on graphs of bounded twin-cover number: for every proper coloring $\varphi$ of the twin cover $G[X]$, we show that the minimum number of colors needed to extend $\varphi$ to all of $G$ is
$K_\varphi=\max_{S\subseteq X}\bigl(|\varphi(S)|+m(S)\bigr)$,
where $m(S)$ is the maximum size of a twin-clique with neighborhood $S$ in $X$. Consequently,
$\chi(G)=\min_{\varphi \text{ proper on }G[X]} K_\varphi$.
This formula expresses the chromatic number in terms of the proper colorings
of $G[X]$ and, for each $S \subseteq X$, the maximum size $m(S)$ of a
twin-clique whose vertices have exactly $S$ as their set of neighbors in $X$.

Taken together, our results show that twin cover is a meaningful parameter for strong CFVC. Algorithmically, we obtain positive evidence that the problem remains tractable beyond when we just consider vertex-cover. Our results isolate the remaining difficulty when solving \textsc{Strong CFVC Number}; to wit., after the twin cover has been fixed, for fixed $k$ \textsc{Strong CFVC Number} remains tractable beyond the vertex-cover setting in which it is hitherto known to be tractable. Structurally, our chromatic formula isolates the remaining difficulty concomitant in the task of computing \textsc{Strong CFVC Number} after twin cover has been fixed: said difficulty is about
finding the smallest integer $d \in \{0,\ldots,|X|\}$ for which $G$ admits of
a strong CFVC coloring with at most $\chi(G)+d$ colors. In this sense, our results place strong CFVC closer to the broader body of coloring problems which exhibit nice behavior when parameterized by the twin-cover, while also highlighting the additional difficulty created by the shortest-path conflict-free requirement \cite{feghali2026parameterized,jansen2019optimal}. We view this as being useful evidence that twin cover is the right framework within which further study of \textsc{Strong CFVC Number} can occur.

\section{Preliminaries and Definitions}

All graphs considered herein are finite, simple, undirected, and connected unless stated otherwise. 

\emph{Graphs, Paths, and Colorings.}
For a graph $G$, we write $V(G)$ and $E(G)$ for its vertex set and edge set, respectively. For a set $X \subseteq V(G)$, we write $G[X]$ for the subgraph of $G$ induced by $X$; we adittionally  write
$G-X := G[V(G)\setminus X]$.
For a vertex $v \in V(G)$, its open neighborhood is
$N_G(v) := \{u \in V(G) : uv \in E(G)\}$,
and its closed neighborhood is $N_G[v] := N_G(v) \cup \{v\}$. For a set $X \subseteq V(G)$, we write
$N_G(X) := \Bigl(\bigcup_{v \in X} N_G(v)\Bigr)\setminus X$.
We omit the subscript $G$ whenever the graph to which the neighborhood applies is clear from the context.

Herein all paths are simple paths. If $P=v_0v_1\cdots v_\ell$ is a path, then the \emph{length} of $P$ is $\ell$. For vertices $u,v \in V(G)$, the \emph{distance} between $u$ and $v$, denoted by $d_G(u,v)$, is the minimum length of a $u$-$v$ path in $G$. A $u$-$v$ path of length $d_G(u,v)$ is a \emph{shortest} $u$-$v$ path.

A \emph{clique} is a set of vertices for which any two vertices therein are adjacent.

A (vertex) \emph{coloring} of $G$ is a function $c \colon V(G) \to \mathbb{N}$. We say that $c$ is a \emph{$k$-coloring} if $c(V(G)) \subseteq [k]$. A coloring $c$ is \emph{proper} if $c(u) \neq c(v)$ for every edge $uv \in E(G)$. The \emph{chromatic number} of $G$, denoted by $\chi(G)$, is the minimum integer $k$ such that $G$ admits a proper $k$-coloring.

We employ image notation in that when we have a coloring $c$ and set of vertices $X \subseteq V(G)$, we use $c(X)$ to denote $\{c(v) : v \in X\}$.

\emph{Strong Conflict-Free Vertex-Connection.} Let $c$ be a coloring of $G$, and let $P=v_0v_1\cdots v_\ell$ be a path in $G$. We say that $P$ is \emph{conflict-free with respect to $c$} if there exists a color $\alpha$ such that $\alpha$ appears on exactly one vertex of $P$, that is,
\[
\bigl|\{i \in \{0,1,\dots,\ell\} : c(v_i)=\alpha\}\bigr| = 1
\]
for some color $\alpha$.

A coloring $c$ of $G$ is a \emph{strong conflict-free vertex-connection coloring}, or briefly a \emph{strong CFVC coloring}, if for every two distinct vertices $u,v \in V(G)$ there exists a conflict-free shortest $u$-$v$ path in $G$.

The \emph{strong conflict-free vertex-connection number} of $G$, denoted by $\operatorname{svcfc}(G)$, is the minimum number of colors in a strong CFVC coloring of $G$.

We consider the following decision problem.

\medskip
\noindent
\textsc{Strong CFVC Number}

\smallskip
\noindent
\emph{Input:} A connected graph $G$ and an integer $k \geq 1$.

\smallskip
\noindent
\emph{Question:} Is $\operatorname{svcfc}(G) \leq k$?

\medskip
\noindent
\textsc{Annotated Strong CFVC Number}

\smallskip
\noindent
\emph{Input:} A connected graph $G$, an integer $k\ge 1$, and a set $X\subseteq V(G)$.

\smallskip
\noindent
\emph{Parameter:} $k+|X|$.

\smallskip
\noindent
\emph{Promise:} $X$ is a twin cover of $G$.

\smallskip
\noindent
\emph{Question:} Is $\operatorname{svcfc}(G)\le k$?

\medskip

We note that the promise can be verified in polynomial time: it suffices to check,
for every edge $uv$ of $G-X$, that $N_G[u]=N_G[v]$.

A color that is not used elsewhere in the coloring under consideration will be called a \emph{fresh color}.

\emph{Twins, Twin Cover, and Twin-Cliques.}
Two distinct vertices $u,v \in V(G)$ are \emph{true twins} if
$N_G[u] = N_G[v]$.
Equivalently, $u$ and $v$ are adjacent and satisfy
$N_G(u)\setminus \{v\} = N_G(v)\setminus \{u\}$. An edge $uv \in E(G)$ is a \emph{twin edge} if its endpoints $u$ and $v$ are true twins. A set $X \subseteq V(G)$ is a \emph{twin cover} of $G$ if every edge of $G-X$ is a twin edge. The \emph{twin-cover number} of $G$, denoted by $\operatorname{tc}(G)$, is the minimum size of a twin cover of $G$. Throughout the paper, whenever a twin cover $X$ is fixed, we write
$t := |X|$. Let $X$ be a fixed twin cover of $G$. A connected component of $G-X$ will be called a \emph{twin-clique}. By the definition of twin cover, every twin-clique is a clique. We give a short proof of this. Let
$uv\in E(G-X)$. Neither $u$ nor $v$ lie in $X$, so $u$ and $v$ have to be true twins. Now let
$v_0v_1\cdots v_r$ be a path contained in one connected component of $G-X$.
Since each of the vertices in $v_i,v_{i+1}$ is joined together using an edge lying in  $G-X$, we have that
$N_G[v_i]=N_G[v_{i+1}]$ for all $i$, and hence, by transitivity, all vertices
$v_0,\dots,v_r$ have the same closed neighborhood in $G$. In particular, they
have the same neighborhood in $X$, and they are pairwise adjacent. Thus every
connected component of $G-X$ is a clique, and all vertices of such a component
have the same neighborhood in $X$; compare also Observation~1 of
\cite{ganian2015improving}. Hence, if $C$ is a twin-clique, we may define
$N_X(C) := N_G(v)\cap X$ for any $v \in C$. For $S \subseteq X$ and $s \in \mathbb{N}$, we write
\[
\mathcal{C}_{S,s}
:=
\{\, C : C \text{ is a twin-clique of } G-X,\ N_X(C)=S,\ |C|=s \,\}.
\]
We refer to the pair $(S,s)$ as the \emph{type} of a twin-clique in $\mathcal{C}_{S,s}$. If $c$ is a coloring and $C$ is a twin-clique, we write
$\operatorname{col}_c(C) := c(C)$
for the set of colors used on $C$. For each $S \subseteq X$, we also define
\[
m(S)
:=
\max\{\, |C| : C \text{ is a twin-clique of } G-X \text{ with } N_X(C)=S \,\},
\]
with the convention that $m(S)=0$ if no such twin-clique exists. If $\varphi$ is a proper coloring of $G[X]$, we write
$\varphi(S) := \{\varphi(x) : x \in S\}$
for the set of colors appearing on $S$, and we define
$K_{\varphi}
:=
\max_{S \subseteq X} \bigl(|\varphi(S)| + m(S)\bigr)$.
The quantity $K_{\varphi}$ will be used later in our treatment of proper colorings of graphs of bounded twin-cover number.

\emph{Parameterized Complexity.} We use standard terminology from parameterized complexity. A parameterized problem is \emph{fixed-parameter tractable} (FPT) if instances of size $n$ with parameter value $\kappa$ can be solved in time
$f(\kappa)\, n^{O(1)}$
for some computable function $f$.

A \emph{kernelization} for a parameterized problem is a polynomial-time algorithm
which maps each instance $(I,\kappa)$ to an equivalent instance $(I',\kappa')$
such that
\[
|I'|+\kappa' \le g(\kappa)
\]
for some computable function $g$. The reduced instance is called a \emph{kernel}.
If $g$ is bounded by a polynomial in $\kappa$, then the problem admits a
\emph{polynomial kernel}.

For \textsc{Annotated Strong CFVC Number}, we call a triple $(G,k,X)$
\emph{valid} if $X$ is a twin cover of $G$; as noted above, this can be checked
in polynomial time. The parameter is
\[
\kappa := k+|X|.
\]
A kernelization for this annotated problem is therefore a polynomial-time algorithm
which maps each valid instance $(G,k,X)$ to an equivalent valid instance
$(G',k',X')$ such that $X'$ is a twin cover of $G'$ and
\[
|V(G')| + k' + |X'| \le h(\kappa)
\]
for some computable function $h$.

Since our reductions are graph reductions, we will usually state kernel bounds in
terms of the number of vertices of the reduced graph. So when we say that 
\textsc{Annotated Strong CFVC Number} has a kernel on at most $h(\kappa)$
vertices, what we mean is that the kernalized instance of reduced size can be chosen such that it satisfies 
$|V(G')|\le h(\kappa)$. 

\section{Basic Structural Properties}

Throughout this section, let $G$ be a connected graph and let
$X \subseteq V(G)$ be a fixed twin cover of size $t$.

\begin{lemma}\label{lem:svcfc-proper}
Every strong CFVC coloring is proper.
\end{lemma}

\begin{proof}
Let $c$ be a strong CFVC coloring of $G$, and let $uv \in E(G)$.
The unique shortest $u$-$v$ path is the edge $uv$ itself.
Hence that path must be conflict-free, which is possible only if $c(u) \neq c(v)$.
Therefore $c$ is proper. \qedbox
\end{proof}

\begin{lemma}\label{lem:shortest-path-twinclique}
Let $P=v_0v_1\cdots v_r$ be a shortest path in $G$, and let $C$ be a twin-clique.
If at least one endpoint of $P$ lies outside $C$, then $P$ contains at most one vertex of $C$.

Consequently, if the two endpoints of $P$ do not belong to the same twin-clique,
then $P$ contains at most one vertex from each twin-clique, meets at most $t+1$
twin-cliques, and has length at most $2t$.
\end{lemma}

\begin{proof}
Suppose first that at least one endpoint of $P=v_0v_1\cdots v_r$ lies outside $C$,
and assume for contradiction that $P$ contains at least two vertices of $C$.
Let $i$ be the first index such that $v_i\in C$, and let $j$ be the last index
such that $v_j\in C$. Then $i<j$.

Because at least one endpoint of $P$ lies outside $C$, we have either $i>0$ or $j<r$.
Assume that $j<r$; the case $i>0$ is symmetric.
Then $v_{j+1}\notin C$.
Since $v_jv_{j+1}\in E(G)$ and all vertices of $C$ have the same neighborhood
outside $C$, it follows that $v_iv_{j+1}\in E(G)$.
Therefore replacing the subpath
\[
v_i v_{i+1}\cdots v_j v_{j+1}
\]
by the single edge $v_iv_{j+1}$ yields a shorter path, a contradiction.
Thus $P$ contains at most one vertex of $C$.

For the consequence, assume that the endpoints of $P$ do not lie in the same
twin-clique.
Then the first part applies to every twin-clique, so $P$ contains at most one
vertex from each twin-clique.
Since distinct twin-cliques are distinct connected components of $G-X$,
there are no edges between them.
Hence any neighbors along $P$ of any vertex of $P$ lying in $V(G)\setminus X$ must be in $X$. 
Therefore
\[
|V(P)\setminus X| \leq |V(P)\cap X| + 1 \leq t+1.
\]

Each vertex of $V(P)\setminus X$ lies in a distinct twin-clique, so $P$ meets at
most $t+1$ twin-cliques.
Also, since $P$ is simple, it contains at most $t$ vertices of $X$.
Hence
\[
|V(P)| = |V(P)\cap X| + |V(P)\setminus X|
\leq t + (t+1) = 2t+1.
\]
Therefore the number of edges of $P$ is at most $2t$, as claimed. \qedbox
\end{proof}

\begin{corollary}\label{cor:long-shortest-path-hits-X}
Every shortest path of length at least $2$ contains a vertex of $X$.
\end{corollary}

\begin{proof}
Let $P$ be a path about which the corollary talks, that is, $P$ is a shortest path of length at least $2$. It cannot be the case that its endpoints belong to the same twin-clique: every twin-clique is a clique. If $P$ avoided $X$, then $P$ would lie entirely in $G \setminus X$, hence inside a single connected component of $G \setminus X$, which is a clique.
But then $P$ would have length at most $1$ --- contradiction!
This proves the claim. \qedbox
\end{proof}

\section{An FPT Kernel for the Combined Parameter $t+k$}

In this section, we first prove a kernel bound for \textsc{Annotated Strong CFVC Number},
in which the input is accompanied by a twin cover $X$ of size $t$.
We then derive, in Corollary~\ref{cor:tc-plus-k-kernel}, the corresponding unanotated
result parameterized by $\operatorname{tc}(G)+k$.
Let $(G,k,X)$ be a valid instance of \textsc{Annotated Strong CFVC Number}, and let $t:=|X|$.

\begin{theorem}\label{thm:kernel-t-plus-k}
For any $(G,k,X)$ of \textsc{Annotated Strong CFVC Number}
with $t:=|X|$,  an equivalent instance
$(G',k',X')$ of  \textsc{Annotated Strong CFVC Number} can be computed in polynomial time such that
\[
|V(G')| \le \max\!\left\{2,\;t+k\,2^t\right\}.
\]
In particular, \textsc{Annotated Strong CFVC Number} admits a kernel and is
fixed-parameter tractable parameterized by $t+k$.
\end{theorem}

\begin{proof}
If some twin-clique has more than $k$ vertices, then by
Lemma~\ref{lem:svcfc-proper} we can output the fixed no-instance
$(K_2,1,V(K_2))$. Otherwise, for each $S\subseteq X$ let $B_S$ be the union of
the twin-cliques $C$ with $N_X(C)=S$. From every nonempty $B_S$ retain
exactly $\min\{k,|B_S|\}$ vertices, including one entire clique of size $m(S)$.
Let $G'$ be the subgraph induced by the retained vertices together with $X$.
The set $X$ remains a twin cover of $G'$.

Fix a proper coloring $\varphi$ of $G[X]$ with colors in $[k]$. We now argue that a set $U_S$
can be equal to the image of $B_S$ under some proper $k$-coloring of $G[X\cup B_S]$ extending $\varphi$ if and only if
\[
U_S\subseteq [k]\setminus\varphi(S)
\quad\text{and}\quad
m(S)\le |U_S|\le |B_S|.
\]
The forward direction of this equivalence is because the largest clique requires $m(S)$ distinct colors, and every vertex
of $B_S$ is adjacent to all of $S$. Conversely,  one can assign each color of $U_S$ to
a different vertex of $B_S$, whereupon one fills each clique with colors of $U_S$ that
are still unused in that clique, proving the desired equivalence.

Now different ones of the sets $B_S$ have no edges between them,
so the sets $U_S$ can be chosen independently. The reduction preserves $m(S)$
and replaces $|B_S|$ by $\min\{k,|B_S|\}$, which leaves the feasible sets
$U_S$ unchanged.

Now define the quotient graph to have a vertex set which is  the result of retaining a single vertex for each of the distinct sets in the  partition of $V(G)$ which consists of the following sets:`each of the nonempty sets $B_S$ and, for $x \in X$, $\{x\}$; and make the edge set of said quotient graph such that two vertices in the quotient graph are adjacent if and only if the sets of which they are a member in $V(G)$ have an edge in $G$ connecting their corresponding parts. Between distinct parts, either all possible edges are present or none are. $G$ and $G'$ have the same quotient graph (up to isomorphism), since $G'$ retains at least one vertex from each part. Because for endpoints in distinct parts of the partition, projecting
a path and suppressing consecutive occurrences of the same part gives a
 walk in the quotient graph of no greater length and because a shortest path in the quotient graph lifts to a path of the same length with any prescribed endpoints in its
first and last parts, we have that distances between endpoints in distinct
parts in the partition equal their distances in the quotient graph and that every shortest path between
such endpoints projects to a shortest path in the quotient graph of the same length. 
Within a part $B_S$, two vertices are either adjacent or have a common
neighbor in $S$. In the latter case $S$ is nonempty, since otherwise their
distinct cliques would be disconnected components of $G$. These observations
also show that $G'$ is connected.

Take proper colorings of $G$ and $G'$ that agree on $X$ and use the same sets
$U_S$. Suppose one is strong CFVC. For endpoints in distinct parts in the
other graph, choose endpoints of the same colors in the corresponding parts
of the strongly colored graph. Project a conflict-free shortest path between
these matching endpoints to the quotient and lift it back, choosing vertices
of the same colors at each step. The lifted path is shortest and has the
same color sequence. For endpoints in the same part, every shortest path is
conflict-free by properness and the distance bound of two. Since every strong
CFVC coloring is proper and the two graphs admit the same choices of
$\varphi$ and $U_S$, the instances $(G,k,X)$ and $(G',k,X)$ are equivalent.

The reduction computes the twin-cliques and groups them by their neighborhoods
in polynomial time. There are at most $2^t$ nonempty groups, each contributing
at most $k$ vertices to $G'$, so $|V(G')|\le t+k\,2^t$. \qedbox
\end{proof}

We now convert the annotated kernel of Theorem~\ref{thm:kernel-t-plus-k} into a standard parameterized result in terms of $\operatorname{tc}(G)+k$.

\begin{corollary}\label{cor:tc-plus-k-kernel}
\textsc{Strong CFVC Number} admits a kernel parameterized by
$\operatorname{tc}(G)+k$. More precisely, every instance $(G,k)$ can be reduced
in polynomial time to an equivalent instance on at most
\[
\max\!\left\{
2,\;
2\,\operatorname{tc}(G)
+
k\,2^{2\,\operatorname{tc}(G)}
\right\}
\]
vertices.
\end{corollary}

\begin{proof}
Let $G^\dagger$ be the graph obtained from $G$ by deleting all twin edges.
By Lemma~2 of \cite{ganian2015improving}, the vertex-cover number of $G^\dagger$
is exactly $\operatorname{tc}(G)$.

Let $M$ be any maximal matching in $G^\dagger$, and define $Y$ to be the set of endpoints of edges of $M$. $Y$ is a vertex cover of $G^\dagger$ because $M$ is maximal. Moreover,
\[
|Y| = 2|M| \leq 2\,\operatorname{vc}(G^\dagger)
= 2\,\operatorname{tc}(G).
\]

Now $Y$ must be a twin cover of $G$ of size at most $2\,\operatorname{tc}(G)$, given that every vertex cover of $G^\dagger$ is a twin cover of $G$. It follows from Theorem~\ref{thm:kernel-t-plus-k} applied to the instance $(G,k,Y)$ that there is a polynomial-time computable equivalent valid (annotated) instance for which the underlying graph of which has at most $\max\!\left\{2,\; |Y|+k\,2^{|Y|}\right\}$
vertices. Forgetting the annotation, we obtain an equivalent instance of
\textsc{Strong CFVC Number} on at most the same number of vertices.

Since $|Y|\le 2\,\operatorname{tc}(G)$, this is at most
\[
\max\!\left\{
2,\;
2\,\operatorname{tc}(G)
+
k\,2^{2\,\operatorname{tc}(G)}
\right\}.
\]

Hence \textsc{Strong CFVC Number} admits a kernel parameterized by
$\operatorname{tc}(G)+k$. \qedbox
\end{proof}

\begin{corollary}

For every fixed integer $k \ge 1$, deciding whether $\operatorname{svcfc}(G) \le k$ is fixed-parameter tractable parameterized by $\operatorname{tc}(G)$.
\end{corollary}

\begin{proof}
By Corollary~\ref{cor:tc-plus-k-kernel}, \textsc{Strong CFVC Number} is
fixed-parameter tractable parameterized by $\operatorname{tc}(G)+k$.
For every fixed integer $k$, this is fixed-parameter tractability
parameterized by $\operatorname{tc}(G)$ alone. \qedbox
\end{proof}

\section{Bounding the Strong CFVC Number via the Chromatic Number}

In this section, we show that for a graph equipped with a twin cover $X$ of size $t$,
the strong conflict-free vertex-connection number differs from the chromatic number
by at most $t$.
We also give an exact formula for $\chi(G)$ in terms of a fixed twin cover.

Throughout this section, let $G$ be a connected graph and let
$X \subseteq V(G)$ be a fixed twin cover of size $t$.

\subsection{A $\chi(G)+t$ Upper Bound}

We begin with a slightly more general statement than the bound we will ultimately use.

\begin{theorem}\label{thm:chi-upper-from-Y}
Let $Y \subseteq X$ be such that every shortest path in $G$ of length at least $3$
contains a vertex of $Y$.
Then
$\operatorname{svcfc}(G) \leq \chi(G)+|Y|$.
\end{theorem}

\begin{proof}
Let $\varphi$ be a proper $\chi(G)$-coloring of $G$.
Define the coloring $\psi$ as follows:
for every vertex $y\in Y$, assign a new color $\alpha_y$, featuring the colors
$\alpha_y$ being pairwise distinct and not appearing on any vertex of $V(G)\setminus Y$;
for every vertex $v\in V(G)\setminus Y$, set $\psi(v):=\varphi(v)$.

$\psi$ uses at most $\chi(G)+|Y|$ colors and is proper: on $V(G)\setminus Y$ it is the same as $\varphi$ and $\varphi$ is proper; also every vertex of $Y$ is colored differently by $\psi$ than any of its other colors.

We claim that $\psi$ is a strong CFVC coloring. To show this, 
let $u,v$ be distinct vertices in $G$. If $d_G(u,v)=1$, then the edge $uv$ is a shortest $u$-$v$ path, and moreover must be a conflict-free edge due to the properness of $\psi$. If $d_G(u,v)=2$, then every shortest $u$-$v$ path has the form $u-w-v$.
Because $\psi$ is proper, the middle vertex $w$ has a color different from the colors
of $u$ and $v$.
Hence the color $\psi(w)$ appears exactly once on the path, and the path is conflict-free.

Finally, assume that $d_G(u,v)\geq 3$.
Let $P$ be any shortest $u$-$v$ path.
By the assumption on $Y$, the path $P$ contains some vertex $y\in Y$.
By construction, the color $\alpha_y=\psi(y)$ appears nowhere else in $G$.
Therefore it appears exactly once on $P$, so $P$ is conflict-free.

Thus every two vertices of $G$ are joined by a conflict-free shortest path under $\psi$.
Hence $\psi$ is a strong CFVC coloring, and
$\operatorname{svcfc}(G) \leq \chi(G)+|Y|$. \qedbox
\end{proof}

We obtain the promised bound by taking $Y=X$.

\begin{corollary}\label{cor:svcfc-between-chi-and-chi-plus-t}
For every connected graph $G$ with twin cover $X$ of size $t$, $\chi(G)\leq \operatorname{svcfc}(G)\leq \chi(G)+t$.
In particular, the additive gap $\operatorname{svcfc}(G)-\chi(G)$ is at most $t$.
\end{corollary}

\begin{proof}
Lemma~\ref{lem:svcfc-proper} gives us the lower bound. For the upper bound, Corollary~\ref{cor:long-shortest-path-hits-X} gives the relevant fact we need to prove it: that
every shortest path of length at least $3$ contains a vertex of $X$.
Hence Theorem~\ref{thm:chi-upper-from-Y} applies with $Y=X$. \qedbox
\end{proof}

\subsection{An Exact Formula for the Chromatic Number}

\begin{theorem}\label{thm:extension-number}
Let $\varphi$ be a proper coloring of $G[X]$.
Among all proper colorings $c$ of $G$ satisfying $c(x)=\varphi(x)$
for every $x\in X$, the minimum value of $|c(V(G))|$ is
\[
K_{\varphi}
=
\max_{S\subseteq X}\bigl(|\varphi(S)|+m(S)\bigr).
\]
\end{theorem}

\begin{proof}
We handle the proof of the lower bound first. Let $c$ be an arbitrary proper coloring satisfying the premise of this theorem statement (i.e., $c$ agrees with $\varphi$ on $X$ and is a proper coloring of $G$). 

Now let $S \subseteq X$. First of all, $|\varphi(S)|$ is a lower bound on the size of the image of $c$ on $V(G)$, because $c$ uses every single color in $\varphi(S)$ in order to color the elements of $S$. We now make the assumption that $m(S)>0$ for the rest of the lower bound proof, as if not,  if $m(S)=0$, then this already gives
$|c(V(G))|\geq |\varphi(S)| = |\varphi(S)|+m(S)$ and our work here is done.

Consider a twin clique $C$ such that $N_X(C)=S$ and $|C|=m(S)$. No color with which any vertex of $S$ is colored by $c$ can also be used by $c$ to color any vertex of $C$, because each vertex in $C$ is incident to each vertex contained in $S$.

 Furthermore, from the size of $C$ we have that $c$ uses exactly $m(S)$ colors (pairwise distinct colors) in coloring the vertices of $C$. This reasoning gives
\[
|c(V(G))|
\geq
|\varphi(S)| + m(S).
\]

From the arbitrariness of $S\subseteq X$ it follows that every proper extension $c$ of $\varphi$
satisfies
\[
|c(V(G))|
\geq
\max_{S\subseteq X}\bigl(|\varphi(S)|+m(S)\bigr)
=
K_{\varphi}.
\]

We now turn attention to proving the upper bound.
Let
\[
K := K_{\varphi}
=
\max_{S\subseteq X}\bigl(|\varphi(S)|+m(S)\bigr).
\]
Because $K\geq |\varphi(X)|$, after renaming colors if necessary we may assume that
$\varphi(X)\subseteq [K]$.

We prove the upper bound by explicitly constructing a proper $K$-coloring of $G$ which extends $\varphi$, that is it maps to the same elements on $X$ but then has domain $V(G)$ while being proper. For said construction, preserve the way in which $\varphi$ colors vertices in $X$. or defining what the coloring does on the other colors, consider first any twin-clique $C$ for which its neighborhood is $S:=N_X(C)$ and its size $s:=|C|$.
Since
$K-|\varphi(S)| \geq m(S)\geq s$,
there are at least $s$ colors in the set $[K]\setminus \varphi(S)$; let the coloring we are here constructing assign any such $s$ colors to the vertices of $C$ in an injective manner. 

Now we show that the constructing coloring is a proper $K$-coloring extending $\varphi$. Properness on $C$ follows from how $C$ is a clique and that its vertices receive distinct colors under said coloring; properness for vertices contained in edges going between $C$ and $S$ is due to how no color with which $C$ is colored belongs to $\varphi(S)$. For properness we need not consider vertices in edges from vertices in $C$ to vertices of $X\setminus S$, as there are no such edges. This case analysis gives a proper coloring of $G$ which  extends $\varphi$ and using only the colors in $[K]$ as for different twin-cliques, their pairwise anticompleteness means that we could choose different colors for distinct twin-cliques and be fine as to satisfying the desired properties. Therefore the minimum number of colors needed by a proper extension of $\varphi$ is at most $K_{\varphi}$.

Combining the two bounds gives the desired claim. \qedbox
\end{proof}

The below corollary is immediate from the above. 

\begin{corollary}\label{cor:chi-formula}
For every connected graph $G$ with fixed twin cover $X$,
\[
\chi(G)
=
\min_{\varphi \text{ proper on } G[X]} K_{\varphi}
=
\min_{\varphi \text{ proper on } G[X]}
\max_{S\subseteq X}\bigl(|\varphi(S)|+m(S)\bigr).
\]
\end{corollary}

\begin{proof}
Every proper coloring of $G$ corresponds to a proper coloring $\varphi$ of  $G[X]$, formed by restricting its domain to $X$. 
By Theorem~\ref{thm:extension-number}, every proper extension of $\varphi$
uses at least $K_{\varphi}$ colors, and there exists a proper extension
using exactly $K_{\varphi}$ colors.
Taking the minimum over all proper colorings $\varphi$ of $G[X]$ gives the desired formula. \qedbox
\end{proof}

\begin{remark}
Corollary~\ref{cor:svcfc-between-chi-and-chi-plus-t} and Corollary~\ref{cor:chi-formula}
show that, once a twin cover $X$ is fixed, the remaining difficulty in determining
$\operatorname{svcfc}(G)$ lies in the bounded additive gap above $\chi(G)$,
which is at most $|X|=t$.
\end{remark}

\bibliographystyle{splncs04}
\bibliography{twin}

\end{document}